\documentclass[12pt,english]{article}
\usepackage[T1]{fontenc}
\usepackage[dvips]{geometry}
\usepackage{amsmath}
\usepackage{amsthm}
\usepackage{amssymb}
\usepackage{setspace}
\usepackage[authoryear]{natbib}
\PassOptionsToPackage{normalem}{ulem}
\usepackage{ulem}
\makeatletter
\theoremstyle{plain}
\newtheorem{assumption}{\protect\assumptionname}
\theoremstyle{plain}
\newtheorem{thm}{\protect\theoremname}
\theoremstyle{plain}
\newtheorem{lem}{\protect\lemmaname}

\usepackage{setspace}

\makeatother

\usepackage{babel}
\providecommand{\assumptionname}{Assumption}
\providecommand{\lemmaname}{Lemma}
\providecommand{\theoremname}{Theorem}

\begin{document}
\title{Sparse network asymptotics for logistic regression}
\author{Bryan S. Graham\thanks{Department of Economics, University of California - Berkeley, 530
Evans Hall \#3380, Berkeley, CA 94720-3888, \uline{e-mail:} $\mathtt{bgraham@econ.berkeley.edu}$,
\uline{web:} $\mathtt{http://bryangraham.github.io/econometrics/}$.
Financial support from NSF Grant SES \#1851647 is gratefully acknowledged.
Some of the results contained in this paper were presented, albeit
in more basic and preliminary forms, at an invited session of the
2018 Latin American Meetings of the Econometric Society, and at a
plenary lecture of the 2019 meetings of the International Association
of of Applied Econometrics. I am thankful to Michael Jansson for several
very helpful conversations and to Konrad Menzel for feedback on the
initial draft. All the usual disclaimers apply.}}
\date{September 2020}
\maketitle
\begin{abstract}
\thispagestyle{empty} 
\end{abstract}
\pagebreak{}
\begin{abstract}
Consider a bipartite network where $N$ consumers choose to buy or
not to buy $M$ different products. This paper considers the properties
of the logistic regression of the $N\times M$ array of ``$i$-buys-$j$''
purchase decisions, $\left[Y_{ij}\right]_{1\leq i\leq N,1\leq j\leq M}$,
onto known functions of consumer and product attributes under asymptotic
sequences where (i) both $N$ and $M$ grow large and (ii) the average
number of products purchased per consumer is finite in the limit.
This latter assumption implies that the network of purchases is \emph{sparse}:
only a (very) small fraction of all possible purchases are actually
made (concordant with many real-world settings). Under sparse network
asymptotics, the first and last terms in an extended Hoeffding-type
variance decomposition of the score of the logit composite log-likelihood
are of equal order. In contrast, under dense network asymptotics,
the last term is asymptotically negligible. Asymptotic normality of
the logistic regression coefficients is shown using a martingale central
limit theorem (CLT) for triangular arrays. Unlike in the dense case,
the normality result derived here also holds under degeneracy of the
network graphon. Relatedly, when there ``happens to be'' no dyadic
dependence in the dataset in hand, it specializes to recently derived
results on the behavior of logistic regression with rare events and
iid data. Sparse network asymptotics may lead to better inference
in practice since they suggest variance estimators which (i) incorporate
additional sources of sampling variation and (ii) are valid under
varying degrees of dyadic dependence.

\smallskip{}

\begin{singlespace}
\noindent \uline{JEL Codes:} C31, C33, C35
\end{singlespace}

\smallskip{}

\begin{singlespace}
\noindent \uline{Keywords}: Networks, Exchangeable Random Arrays,
Dyadic Clustering, Sparse Networks, Logistic Regression, Rare Events,
Marginal Effects
\end{singlespace}

\thispagestyle{empty} 

\pagebreak{}

\setcounter{page}{1}
\end{abstract}
Let $i=1,\ldots,N$ index a random sample of consumers and $j=1,\ldots,M$
a random sample of products. For each consumer-product pair $ij$
we observe $Y_{ij}=1$ if consumer $i$ purchases product $j$ and
$Y_{ij}=0$ otherwise. Let $W_{i}$ be a vector of observed consumer
attributes, $X_{j}$ a vector of product attributes and $n=M+N$ the
total number of sampled consumers and products. The conditional probability
that consumer $i$ buys product $j$ is given by
\begin{equation}
\Pr\left(\left.Y_{ij}=1\right|W_{i},X_{j}\right)=e\left(\alpha_{0,n}+Z_{ij}'\beta_{0}\right),\label{eq: purchase_probability}
\end{equation}
where $Z_{ij}\overset{def}{\equiv}z\left(W_{i},X_{j}\right)$ is a
vector of known functions of $W_{i}$ and $X_{j}$, $\alpha_{0,n}$
an ``intercept'' parameter (which may vary with $n$), $\beta_{0}$
a vector of fixed ``slope'' parameters, and $e\left(\cdot\right)$
a known increasing function mapping the real line into the unit interval.
Below I will emphasize the logit case with $e\left(v\right)=\exp\left(v\right)/\left[1+\exp\left(v\right)\right]$;
this case is convenient and dominates empirical work, but nothing
which follows hinges essentially upon it.

I am interested in settings where both the number of consumers, $N$,
and the number of products, $M$, are very large. To motivate this
focus, consider a large book retailer. Such a retailer may service
many customers and also stock many books. Let $x$ be the attribute
vector associated with a newly released book, $\hat{\theta}_{n}=\left(\hat{\alpha}_{n},\hat{\beta}'\right)'$
estimates of the parameters in \eqref{eq: purchase_probability},
constructed from some training sample, and $\hat{e}_{i}\left(x\right)=e\left(\hat{\alpha}_{n}+z\left(W_{i},x\right)'\hat{\beta}\right)$
the predicted probability that agent $i$ purchases a book of type
$X_{j}=x$. With this knowledge the retailer might use
\begin{equation}
\hat{\gamma}_{n}\left(x\right)=\sum_{i=1}^{N}\hat{e}_{i}\left(x\right)\label{eq: predicted_sales}
\end{equation}
to predict total unit sales for the new book. This prediction could
be useful for making wholesale purchase decisions. Other objects of
interest include various average partial effects \citep[e.g.,][]{Chamberlain_HBE84,Wooldridge_IIEM05}.

In this paper I present a method of estimating the coefficient vector
$\theta_{0,n}=\left(\alpha_{0,n},\beta_{0}'\right)'$ as well as one
for conducting inference on it. I also explore estimation and inference
for aggregate effects, like \eqref{eq: predicted_sales}, as well
as for average effects. The econometric framework outlined below is
designed to accommodate two peculiarities of the setting described
above that appear to be important in practice and also consequential
for inference.

First, consider predicting whether randomly sampled consumer $i$
purchases book $j$, say \emph{The Clue in the Crossword Cipher},
the forty-fourth novel in the celebrated Nancy Drew mystery series.
Knowledge of the frequency with which other consumers $k=1,\ldots,i-1,i+1,\ldots,N$
purchase book $j$ will generally alter the econometrician's prediction
of whether $i$ also purchases book $j$. That is, for any $k\neq i$,
\[
\Pr\left(\left.Y_{ij}=1\right|Y_{kj}=1\right)>\Pr\left(Y_{ij}=1\right)
\]
or that $Y_{i_{1}j_{1}}$and $Y_{i_{2}j_{2}}$ will covary whenever
the two transactions correspond to a common book (such that $j_{1}=j_{2}$).

Similarly, if the econometrician knew that consumer $i$ was a frequent
book buyer, she might conclude that this consumer is also more likely
to purchase some other book (relative to the average consumer). That
is $Y_{i_{1}j_{1}}$and $Y_{i_{2}j_{2}}$ will also covary whenever
the transactions correspond to a common buyer (such that $i_{1}=i_{2}$).

Importantly, dependence across $Y_{i_{1}j_{1}}$ and $Y_{i_{2}j_{2}}$
when $\left\{ i_{1},j_{j}\right\} $ and $\left\{ i_{2},j_{2}\right\} $
share a common buyer or book index may hold even conditional on observed
consumer, $W_{i}$, and product attributes, $X_{j}$. Some consumers
may have latent attributes (i.e., not contained in $W_{i}$) which
induce them to buy many books and some books may be especially popular
(for reasons not captured adequately by $X_{j}$). It might be, for
example, that
\[
\Pr\left(\left.Y_{ij}=1\right|Y_{kj}=1,W_{i},W_{k},X_{j}\right)>\Pr\left(\left.Y_{ij}=1\right|W_{i},X_{j}\right).
\]

The structured form of dependence across the elements of $\left[Y_{ij}\right]_{1\leq i\leq N,1\leq j\leq M}$
described above is a feature of separately exchangeable random arrays
\citep{Aldous_JMA81,Hoover_WP79}. The inferential implications of
such dependence, in the context of subgraph counts, were first considered
by \citet{Holland_Leinhardt_SM76} almost fifty years ago. \citet{Bickel_et_al_AS11}
make an especially important recent contribution in this area. In
the context of regression models, the inferential implications of
dyadic dependence have been considered by, among others, \citet{Fafchamp_Gubert_JDE07},
\citet{Cameron_Miller_WP14}, \citet{Aronow_et_al_PA17}, \citet{Graham_Book_DR_Chap2020},
\citet{Davezies_et_al_AS20} and \citet{Menzel_arXiv17} (see \citet[Section 4]{Graham_HBE2020}
for a review and references). Dyadic dependence will generate distinct
issues here.

The second peculiarity explored here is suggested by the observation
that, even when presented with the opportunity to purchase many books,
the typical customer will only purchase a few. Similarly, a retailer
only sells a few copies of most titles in a given year. Put differently
personal libraries are generally small and the market share of most
books is, for all practical purposes, infinitesimally small. These
observations have implications for what types of asymptotic approximations
are likely to be useful in practice. In this paper I consider sequences
where both $N$ and $M$ grow at the same rate such that, recalling
that $n=M+N$, 
\[
M/n\rightarrow\phi\in\left(0,1\right)
\]
as $n\rightarrow\infty$. 

Let $\rho_{0,n}\overset{def}{\equiv}\mathbb{E}\left[e\left(\alpha_{0,n}+Z_{ij}'\beta_{0}\right)\right]$
be probability that a randomly sampled consumer purchases a randomly
sampled book. The average number of books purchased by the average
consumer is then
\begin{equation}
\lambda_{0,n}^{c}\overset{def}{\equiv}M\rho_{0,n}.\label{eq: average_degree}
\end{equation}
In network parlance $\lambda_{0,n}^{c}$ corresponds to average consumer
degree. If $\rho_{0,n}$ is bounded away from zero, then $\lambda_{0,n}^{c}\rightarrow\infty$
as $N,M\rightarrow\infty$. This implies that the number of actual
book purchases and the number of possible book purchases should be
of equal order. This not true in practice.\footnote{There are tens of millions of print titles available on, for example,
Amazon, even consumers who buys hundreds of books in a year are completing
only very small fraction of all possible purchases.} To develop a distribution theory which is concordant with the empirical
regularity that consumers only purchase a small number of books (and,
similarly, that retailers only sell a small number of copies of any
given title) I let $\alpha_{0,n}\rightarrow-\infty$ at a rate which
ensures that $\lambda_{0,n}^{c}$ converges to a non-zero and bounded
constant $0<\lambda_{0}^{c}<\infty$ as $N,M\rightarrow\infty$. In
language of networks I consider bi-partite graphs which are \emph{sparse}.

The asymptotic analysis in this paper is, to my knowledge, novel,
but it does connect with two important areas of prior research by
others. The first is the literature on subgraph counts and dyadic
regression cited above. However, with the partial exception of Bickel
et al.'s \citeyearpar{Bickel_et_al_AS11} analysis of acyclic subgraph
counts, this work has been, starting with \citet{Holland_Leinhardt_SM76},
limited to to dense networks.\footnote{\citet{Graham_EM17} and \citet{Jochmans_JBES18} also considered
regression in the context of graphs which are sparse in the limit.
Both these papers utilize conditional likelihood type ideas; this
has the effect of ``conditioning away'' some of the dependence which
is central to the analysis below.} The second connection is to the literature on ``rare events'' analysis
\citep[e.g.,][]{King_Zeng_PA2001}; an area of special concern in
political science and epidemology, but also increasingly relevant
in economics (especially in the era of ``Big Data''). An interesting
feature of Theorem \ref{thm: asymptotic_distribution} below is that
it contains Wang's \citeyearpar{HaiYing_ICML2020} recent result for
logistic regression with rare events and iid data as a special case.

The formal analysis of this paper is confined to bipartite networks,
but adapting it to directed and/or undirected networks would be straightforward.
Several consumer demand settings might be appropriately modeled with
the methods described in this paper. Examples include (i) the listening
behavior of streaming music service customers and (ii) the purchase
behavior of big box store customers. A limitation vis-a-vis these
applications, is that the basic set-up explored here is not useful
for understanding complementary and substitution patterns across products
\citep[cf.,][]{Lewbel_Nesheim_Cemmap2019}. Other possible applications
include modeling plant locations in an industry where firms typically
operate multiple plants (here $i=1,\ldots,N$ would index firms and
$j=1,\ldots,M$ locations; see \citet{KPMG_CAReport2016}). Other
many-to-many matching problems that have drawn economists' interest
include (i) bank-firm lending relationships \citep[e.g.,][]{Marotta_et_al_PlosOne2015},
(ii) the matching of venture capital with start-ups \citep[e.g.,][]{Bengtsson_Hsu_JBV2015},
and (iii) supply chain settings with strong bi-partite structure (e.g.,
automakers and parts supplies as in \citet{Fox_QE2018}). When $i=1,\ldots,N$
and $j=1,\ldots,M$ index the same units with $M=N$, applications
include the modeling of ``rare events'' in international relations
data, such as interstate wars \citep[e.g.,][]{King_Zeng_PA2001}.
More generally, the methods developed in this paper, with minimal
adaptation, can be used for link prediction in any sparse network
setting: bi-partite, directed or undirected.\footnote{The methods outlined here are not appropriate for use in one-to-one
matching settings.} There are numerous applications of link prediction in the other social
sciences, the bench sciences as well as in industry.

In what follows random variables are denoted by capital Roman letters,
specific realizations by lower case Roman letters and their support
by blackboard bold Roman letters. That is $Y$, $y$ and $\mathbb{Y}$
respectively denote a generic random draw of, a specific value of,
and the support of, $Y$. A ``0'' subscript on a parameter denotes
its population value and may be omitted when doing so causes no confusion.
In what follows I use graph, network and purchase graph to refer to
$\mathbf{Y}\overset{def}{\equiv}\left[Y_{ij}\right]_{1\leq i\leq N,1\leq j\leq M}$.
All graph theory terms and notation used below are standard \citep[e.g.,][]{Chartand_Zhang_GT2012}.

\section{\label{sec: population}Population and sampling assumptions}

Let $i\in\mathbb{N}$ index\emph{ consumers} in an infinite population
of interest. Associated with each consumer is the vector of observed
attributes $W_{i}\in\mathbb{W}=\left\{ w_{1},\ldots,w_{J}\right\} .$
Let $j\in\mathbb{M}$ index \emph{products} in a second infinite population
of interest. The model is a two population one \citep[cf.,][]{Graham_Imbens_Ridder_JBES18}.
Associated with each product is the vector of characteristics $X_{i}\in\mathbb{X}=\left\{ x_{1},\ldots,x_{K}\right\} $.
The finite support assumption on $\mathbb{W}$ and $\mathbb{X}$ is
not essential, but simplifies the discussion of exchangeability below. 

Let $\sigma_{w}:\mathbb{N}\rightarrow\mathbb{N}$ be a permutation
of a finite number of consumer indices which satisfies the restriction
\begin{equation}
\left[W_{\sigma_{w}\left(i\right)}\right]_{i\in\mathbb{N}}=\left[W_{i}\right]_{i\in\mathbb{N}}.\label{eq: permutation_definition}
\end{equation}
Restriction \eqref{eq: permutation_definition} implies that $\sigma_{w}$
only permutes indices across observationally identical consumers (i.e.,
with the same values of $W$). Let $\sigma_{x}:\mathbb{M}\rightarrow\mathbb{M}$
be an analogously constrained permutation of a finite number of product
indices. Adapting the terminology of \citet{Crane_Towsner_JSL18},
I assume that the purchase graph is $W$-$X$-\emph{exchangeable}
\begin{equation}
\left[Y_{\sigma_{w}\left(i\right)\sigma_{x}\left(j\right)}\right]_{i\in\mathbb{N},j\in\mathbb{M}}\overset{D}{=}\left[Y_{ij}\right]_{i\in\mathbb{N},j\in\mathbb{M}}.\label{eq: WXexchangeable}
\end{equation}
Here $\overset{D}{=}$ denotes equality of distribution. One way to
think about \eqref{eq: WXexchangeable} is as a requirement that any
probability law for $\left[Y_{ij}\right]_{i\in\mathbb{N},j\in\mathbb{M}}$
should attach equal probability to all purchase graphs which are isomorphic
as vertex-colored graphs. Here I associate $W_{i}$ and $X_{j}$ with
the color of the corresponding consumer and product vertices in the
overall purchase graph. Virtually all single-population micro-econometric
models assume that agents are exchangeable, restriction \eqref{eq: WXexchangeable}
extends this idea to the two-population setting considered here. Our
probability law for the model should not change if we re-label observationally
identical units.

\subsection*{Graphon}

It is well-known that exchangeability implies restrictions on the
structure of dependence across observations in the cross-section setting
\citep[e.g.,][]{deFinetti_AN1931}. \citet{Aldous_JMA81}, \citet{Hoover_WP79}
and \citet{Crane_Towsner_JSL18} showed that exchangeable random \emph{arrays}
also exhibit a special dependence structure. Let $\mu$, $\left\{ \left(W_{i},A_{i}\right)\right\} _{i\geq1}$,
$\left\{ \left(X_{j},B_{j}\right)\right\} _{j\geq1}$ and $\left\{ V_{ij}\right\} _{i\geq1,j\geq1}$
be sequences of i.i.d. random variables, additionally independent
of one another, and consider the purchase graph $\left[Y_{ij}^{*}\right]_{i\in\mathbb{N},j\in\mathbb{M}}$,
generated according to
\begin{equation}
Y_{ij}^{*}=h\left(\mu,W_{i},X_{j},A_{i},B_{j},V_{ij}\right)\label{eq: graphon}
\end{equation}
with $h:\left[0,1\right]\times\mathbb{W}\times\mathbb{X}\times\left[0,1\right]^{2}\rightarrow\left\{ 0,1\right\} $
a measurable function, henceforth referred to as a \emph{graphon}
(we can normalize $\mu$, $A_{i}$, $B_{j}$ and $V_{ij}$ to have
support on the unit interval, uniformly distributed, without loss
of generality). 

The results of \citet{Crane_Towsner_JSL18}, which extend the earlier
work of \citet{Aldous_JMA81} and \citet{Hoover_WP79}, show that,
for any $W$-$X$-\emph{exchangeable }random array $\left[Y_{ij}\right]_{i\in\mathbb{N},j\in\mathbb{M}}$,
there exists another array $\left[Y_{ij}^{*}\right]_{i\in\mathbb{N},j\in\mathbb{M}}$,
generated according to \eqref{eq: graphon}, such that the two arrays
have the same distribution. An implication of this result is that
we may use \eqref{eq: graphon} as a nonparametric data generating
process for $\left[Y_{ij}\right]_{i\in\mathbb{N},j\in\mathbb{M}}$. 

Inspection of \eqref{eq: graphon} indicates that exchangeability
implies a particular pattern of dependence across the elements of
$\left[Y_{ij}\right]_{i\in\mathbb{N},j\in\mathbb{M}}$. In particular
$Y_{i_{1}j_{1}}$ and $Y_{i_{2}j_{2}}$ may covary whenever $i_{1}=i_{2}$
or $j_{1}=j_{2}$; this covariance may be present even conditional
on consumer and product attributes. This is, of course, precisely
the dependence structure discussed in the introduction.

\subsection*{Sampling process}

Let $\mathbf{Y}=\left[Y_{ij}\right]_{1\leq i\leq N,1\leq j\leq M}$
be the observed $N\times M$ matrix of consumer purchase decisions.
Let $\mathbf{W}$ and $\mathbf{X}$ be the associated matrices of
consumer and product regressors. I assume that $\mathbf{Y}$ is the
adjacency matrix associated with the subgraph induced by a random
sample of consumers and products from a $W$-$X$-\emph{exchangeable
}infinite population graph. Let $G_{\infty,\infty}$ denote this population
network. Associated with this network is some graphon \eqref{eq: graphon}.
Let $\mathcal{V}_{c}$ and $\mathcal{V}_{p}$ denote the set of consumers
and products randomly sampled by the econometrician from $G_{\infty,\infty}$.
We have $\mathbf{Y}$ equal to the adjacency matrix of the network:
\begin{equation}
G_{N,M}=G_{\infty,\infty}\left[\mathcal{V}_{c},\mathcal{V}_{p}\right].\label{eq: data_as_induced_subgraph}
\end{equation}
An implication of \eqref{eq: WXexchangeable}, \eqref{eq: graphon}
and \eqref{eq: data_as_induced_subgraph} is that we may proceed `as
if' the adjacency matrix in hand was generated according to
\[
Y_{ij}=h\left(\mu,W_{i},X_{j},A_{i},B_{j},V_{ij}\right)
\]
for $i=1,\ldots,N$ and $j=1,\ldots,M$. The marginal probability
of the event, random consumer $i$, purchases random product $j$,
is thus
\begin{equation}
\rho_{0}=\mathbb{E}\left[h\left(\mu,W_{i},X_{j},A_{i},B_{j},V_{ij}\right)\right].\label{eq: density}
\end{equation}

Let $\left\{ G_{N,M}\right\} $ be a sequence of networks indexed
by, respectively, the cardinality of the sampled consumer and product
index sets, $N=\left|\mathcal{V}_{c}\right|$ and $M=\left|\mathcal{V}_{p}\right|$.
The average number of products purchased per consumer, or \emph{average
consumer degree},
\begin{equation}
\lambda_{0}^{c}=M\rho_{0}\label{eq: avg_consumer_degree}
\end{equation}
will diverge as $M\rightarrow\infty$ when $\rho_{0}>0$. Likewise
the average number of times a given product is purchased, or \emph{average
product degree},
\begin{equation}
\lambda_{0}^{p}=N\rho_{0}\label{eq: avg_product_degree}
\end{equation}
will also diverge as $N\rightarrow\infty$. A consequence of this
divergence is that the number of possible purchases, and the number
of actual purchases, will be of equal order. In practice, however,
only a small fraction of all possible purchases are made. To capture
this feature of the real world in our asymptotic approximations requires
a slightly more elaborate thought experiment; which I outline next.

Instead of considering a sequence of graphs sampled from a \emph{fixed
}population, I consider a sequence of graphs sampled from a corresponding
\emph{sequence }of populations. The sequence of networks $\left\{ G_{N,M}\right\} $
is one where both $N$ and $M$ grow at the same rate such that, recalling
that $n=M+N$, 
\[
M/n\rightarrow\phi\in\left(0,1\right)
\]
as $n\rightarrow\infty$. For each $N,M$ the graphon describing the
infinite population sampled from is
\begin{equation}
Y_{ij}=h_{N,M}\left(\mu,W_{i},X_{j},A_{i},B_{j},V_{ij}\right).\label{eq: graphon_N_M}
\end{equation}
 This sequence of graphons/populations $\left\{ h_{N,M}\right\} $
has the property that network \emph{density} 
\[
\rho_{0,N,M}=\mathbb{E}_{N,M}\left[h_{N,M}\left(\mu,W_{i},X_{j},A_{i},B_{j},V_{ij}\right)\right]
\]
may approach zero as $n\rightarrow\infty$. Under this setup the order
of $\lambda_{0,N,M}^{c}=M\rho_{0,N,M}$ and $\lambda_{0,N,M}^{p}=N\rho_{0,N,M}$
will depend upon the speed with which $\rho_{0,N,M}$ approaches zero
as $n\rightarrow\infty$. Here I use the notation $\mathbb{E}_{N,M}\left[\cdot\right]$
to emphasize that the probability law used to compute expectations
may vary with the sample size.

As in other exercises in alternative asymptotics, indexing the population
data generating process by the sample size is not meant to capture
a literal feature of how the data are generated, rather it is done
so that the limiting properties of the model share important features
-- in this case sparseness -- with the actual finite network in
hand. In other settings such an approach has led to more useful asymptotic
approximations, a premise I maintain here \citep[e.g.,][]{Staiger_Stock_Em1997}.

\section{\label{sec: logit}Composite likelihood estimator}

The estimation target is the regression function of $Y_{ij}$ given
$X_{i}$ and $W_{j}$. This is a predictive function and may, or may
not, have structural economic meaning as well (see \citet[Sections 4-5]{Graham_HBE2020}).
I assume that this regression function takes the parametric form
\begin{equation}
e\left(\alpha_{0,n}+Z_{ij}'\beta_{0}\right)=\frac{\exp\left(\alpha_{0,n}+Z_{ij}'\beta_{0}\right)}{1+\exp\left(\alpha_{0,n}+Z_{ij}'\beta_{0}\right)}\label{eq: marginal_cond_prob}
\end{equation}
where $Z_{ij}\overset{def}{\equiv}z\left(W_{i},X_{j}\right)$ is a
finite vector of known functions of $W_{i}$ and $X_{j}$. It would
be interesting to extend what follows to semiparametric regression
models, but this is not done here.

Assumption \ref{ass: sampling} formalizes the population and sampling
set-up of the previous section.
\begin{assumption}
\label{ass: sampling}\textsc{(Sampling) }\textup{\emph{The sampled
network is the one induced by a random sample of $N$ consumers and
$M$ products drawn from the nodes of the infinite}}\emph{ $W$-$X$-exchangeable
bipartite random array $\left[Y_{ij}\right]_{i\in\mathbb{N},j\in\mathbb{M}}$}\textup{\emph{
with graphon \eqref{eq: graphon_N_M}}}\emph{;} $N$ and $M$ grow
such that, for $n=M+N,$ 
\[
M/n\rightarrow\phi\in\left(0,1\right)
\]
as $n\rightarrow\infty$. 
\end{assumption}
To allow the probability of making a purchase decline with $n$, let
$\alpha_{0,n}=\ln\left(\alpha_{0}/n\right).$ This gives, after some
manipulation,
\begin{align}
e\left(\alpha_{0,n}+Z_{ij}'\beta_{0}\right) & =\frac{\frac{\alpha_{0}}{n}\exp\left(Z_{ij}'\beta_{0}\right)}{1+\frac{\alpha_{0}}{n}\exp\left(Z_{ij}'\beta_{0}\right)}\label{eq: sparse_prob}
\end{align}
and hence an expression for average consumer degree \eqref{eq: average_degree}
of, recalling that $M/n\approx\phi$, 
\begin{align*}
\lambda_{0,n}^{c} & =\alpha_{0}\phi\mathbb{E}\left[\exp\left(Z_{ij}'\beta_{0}\right)\right]+O\left(n^{-1}\right)
\end{align*}
which converges to a bounded constant as $n\rightarrow\infty$ as
long as $\mathbb{E}\left[\exp\left(Z_{ij}'\beta_{0}\right)\right]<\infty$.
By allowing $\alpha_{0,n}\rightarrow-\infty$ as $n\rightarrow\infty$
we ensure that, in the limit, the bipartite graph $\mathbf{Y}=\left[Y_{ij}\right]$
is sparse. Similar devices are used by \citet{Owen_JMLR2007} and
\citet{HaiYing_ICML2020} to model ``rare events'' in cross-sectional
binary outcome data.
\begin{assumption}
\label{ass: logit_regression}\textsc{(Logit Regression Function)
}\textup{\emph{The mean regression function (CEF) $\mathbb{E}_{N,M}\left[\left.Y_{ij}\right|W_{i},X_{j}\right]$
belongs to the parametric family \eqref{eq: marginal_cond_prob} with
$\alpha_{n}=\ln\left(\alpha/n\right)$, $\theta=\left(\alpha,\beta'\right)'\in\mathbb{A}\times\mathbb{B}=\Theta$,
$\mathbb{A}$ and $\mathbb{B}$ compact, and }}\emph{$Z_{ij}\in\mathbb{Z}$}
with $\mathbb{\mathbb{Z}}$ a compact subset of $\mathbb{R}^{\dim\left(\mathbb{Z}_{ij}\right)}$.
The true parameter $\theta_{0}=\left(\alpha_{0},\beta_{0}'\right)'$
lies in the interior of the parameter space.
\end{assumption}
The compact support assumption on $Z_{ij}$ is not essential, but
simplifies the proofs. Here I focus on estimation of $\theta_{n}=$$\left(\alpha_{n},\beta'\right)'$
with $\alpha_{n}=\ln\left(\alpha/N\right)$. Observe that $\theta=\left(\alpha,\beta'\right)'$
does not vary with $n$, while $\theta_{n}$ does. Define $\theta_{0,n}=\left(\alpha_{0,n},\beta_{0}'\right)'$
with $\alpha_{0,n}=\ln\left(\alpha_{0}/n\right)$. Let $\bar{\alpha}=\sup\mathbb{A}$,
the parameter space for $\theta_{n}$, $\Theta_{n}=\left(-\infty,\ln\left(\bar{\alpha}\right)\right]\times\mathbb{B}$,
is the one induced by $\Theta=\mathbb{A}\times\mathbb{B}$ and the
mapping from $\theta$ to $\theta_{n}$. 

To estimate $\theta_{0,n}$we maximize the composite log-likelihood
function
\[
\hat{\theta}_{n}=\arg\underset{\theta\in\Theta_{n}}{\max}\thinspace L_{n}\left(\theta\right)
\]
with $L_{n}\left(\theta\right)\overset{def}{\equiv}\frac{1}{NM}\sum_{i=1}^{N}\sum_{j=1}^{M}l_{ij}\left(\theta\right)$
and 
\[
l_{ij}\left(\theta\right)=\left(2Y_{ij}-1\right)R_{ij}'\theta-\ln\left[1+\exp\left[\left(2Y_{ij}-1\right)R_{ij}'\theta\right]\right]
\]
the logit kernel function and $R_{ij}=\left(1,Z_{ij}'\right)'$. Observe
that $\hat{\theta}_{n}$ is simply the coefficient vector associated
with a logistic regression of $Y_{ij}$ onto a constant and $Z_{ij}$
using all $NM$ dyads in the network. Although, by virtue of Assumption
\ref{ass: logit_regression} above, $L_{n}\left(\theta\right)$ correctly
represents the marginal (conditional) probability of $Y_{ij}$ for
each element of $\mathbf{Y}$, it does not accurately reflect the
dependence structure across these elements; hence the term ``composite
likelihood''. See \citet{Lindsay_CM88} an introduction to estimation
by composite likelihood and \citet{Graham_HBE2020} for discussion
in the contexts of network model estimation.

\section{Sparse network asymptotics}

If $\alpha_{0,n}$ equals a fixed constant, then $\rho_{0,n}$ - network
density -- will also be fixed such that the network will be dense
in the limit. The limit distribution of $\hat{\theta}_{n}$ under
such ``dense network asymptotics'' was derived by \citet{Graham_HBE2020}.
More general results for dyadic M-estimators under dense network asymptotics,
including results on the bootstrap, can be found in \citet{Menzel_arXiv17}
and \citet{Davezies_et_al_AS20}. None of these results apply here.
To derive a result that does apply, begin with the mean value expansion
\[
\sqrt{n}\left(\hat{\theta}_{n}-\theta_{0,n}\right)=\left[nH_{n}\left(\bar{\theta}_{n}\right)\right]^{+}\times n^{3/2}S_{n}\left(\theta_{0,n}\right).
\]
where

\begin{equation}
S_{n}\left(\theta\right)=\frac{1}{NM}\sum_{i=1}^{N}\sum_{j=1}^{M}s_{ij}\left(\theta\right),\label{eq: score}
\end{equation}
with $s_{ij}\left(\theta\right)=\frac{\partial l_{ij}\left(\theta\right)}{\partial\theta}=\left(Y_{ij}-e_{ij}\left(\theta\right)\right)R_{ij}$
and $e_{ij}\left(\theta\right)=e\left(\alpha+Z_{ij}'\beta\right)$,
corresponds to the score vector of the composite likelihood and
\begin{equation}
H_{n}\left(\theta\right)=\frac{1}{NM}\sum_{i=1}^{N}\sum_{j=1}^{M}\frac{\partial^{2}l_{ij}\left(\theta\right)}{\partial\theta\partial\theta'}\label{eq: hessian}
\end{equation}
the associated Hessian matrix. Here $\bar{\theta}_{n}$ is a mean
value between $\theta_{0,n}$ and $\hat{\theta}_{n}$ which may vary
from row to row. 

Lemma \ref{lem: Hessian_Convergence}, stated and proved in Appendix
\ref{app: lemmas}, shows that, after re-scaling by $n$, that $nH_{n}\left(\theta\right)$
converges uniformly to 

\begin{equation}
\Gamma\left(\theta\right)=-\alpha\mathbb{E}\left[\exp\left(Z_{12}'\beta\right)\left(\begin{array}{cc}
1 & Z_{12}'\\
Z_{12} & Z_{12}Z_{12}'
\end{array}\right)\right].\label{eq: GAMMA}
\end{equation}
An intuition for why $H_{n}\left(\theta\right)$ needs to be rescaled
to ensure convergence is that, under sparse network asymptotics, information
accrues at a slower rate: the effective sample size is not $NM=\left(n^{2}\right)$,
but rather $O\left(n\right)$. I return to this point briefly at the
end of the paper.
\begin{assumption}
\label{ass: identification}\textsc{(Identification) }The matrix $\Gamma_{0}\overset{def}{\equiv}\Gamma\left(\theta_{0}\right)$
is of full rank.
\end{assumption}
Assumption \ref{ass: identification} is a standard identification
condition (see, for example, \citet[p. 270]{Amemiya_AE85}). This
assumption, in conjunction with Lemma \ref{lem: Hessian_Convergence},
gives the linear approximation
\[
\sqrt{n}\left(\hat{\theta}_{n}-\theta_{n}\right)=\Gamma_{0}^{-1}\times n^{3/2}S_{n}\left(\theta_{0,n}\right)+o_{p}\left(1\right).
\]
To derive the limit distribution of $\sqrt{n}\left(\hat{\theta}_{n}-\theta_{n}\right)$
I show that the distribution $n^{3/2}S_{n}\left(\theta_{0,n}\right)$
is well-approximated by a Gaussian random variable. The main tool
used is a martingale CLT for triangular arrays. That the variance
stabilizing rate for $S_{n}\left(\theta_{0,n}\right)$ is $n^{3/2}$,
like the need to rescale the Hessian, is non-standard. The need to
``blow up'' $S_{n}\left(\theta_{0,n}\right)$ at a faster than $\sqrt{n}$
rate is a consequence of the fact that the summands in $S_{n}\left(\theta_{0,n}\right)$
are $O\left(n^{-1}\right)$ since $\alpha_{0,n}\rightarrow-\infty$
as $n\rightarrow\infty$.

A detailed proof of Theorem \ref{thm: asymptotic_distribution}, stated
below, is provided in Appendix \ref{app: proof_of_main_result}. Here
I outline the main arguments. Begin with the following three part
decomposition of the score vector
\begin{align}
S_{n}\left(\theta\right) & =U_{1n}\left(\theta\right)+U_{2n}\left(\theta\right)+V_{n}\left(\theta\right)\label{eq: score_decomposition}
\end{align}
where
\begin{align}
U_{1n}\left(\theta\right)= & \frac{1}{N}\sum_{i=1}^{N}\bar{s}_{1i}^{c}\left(\theta\right)+\frac{1}{M}\sum_{j=1}^{M}\bar{s}_{1j}^{p}\left(\theta\right)\label{eq: U1n}\\
U_{2n}\left(\theta\right)= & \frac{1}{NM}\sum_{i=1}^{N}\sum_{j=1}^{M}\left\{ \bar{s}_{ij}\left(\theta\right)-\bar{s}_{1i}^{c}\left(\theta\right)-\bar{s}_{1j}^{p}\left(\theta\right)\right\} \label{eq: U2n}\\
V_{n}\left(\theta\right)= & \frac{1}{NM}\sum_{i=1}^{N}\sum_{j=1}^{M}\left\{ s_{ij}\left(\theta\right)-\bar{s}_{ij}\left(\theta\right)\right\} \label{eq: Vn}
\end{align}
with $\bar{s}_{ij}\left(\theta\right)=\bar{s}\left(W_{i},X_{j},A_{i},B_{j};\theta\right)$
with $\bar{s}\left(w,x,a,b;\theta\right)=\mathbb{E}\left[\left.s_{ij}\left(\theta\right)\right|W_{i}=w,X_{j}=x,A_{i}=a,B_{j}=b\right]$
and
\begin{align*}
\bar{s}_{1i}^{c}\left(\theta\right)= & \bar{s}_{1}^{c}\left(W_{i},A_{i};\theta\right)\\
\bar{s}_{1j}^{p}\left(\theta\right)= & \bar{s}_{1}^{p}\left(X_{j},B_{j};\theta\right)
\end{align*}
with $\bar{s}_{1}^{c}\left(w,a;\theta\right)=\mathbb{E}\left[\bar{s}\left(w,X_{j},a,B_{j};\theta\right)\right]$
and $\bar{s}_{1}^{p}\left(x,b;\theta\right)=\mathbb{E}\left[\bar{s}\left(W_{i},x,A_{i},b;\theta\right)\right]$.

Decomposition \eqref{eq: score_decomposition} also features in \citet{Graham_Book_DR_Chap2020}
and \citet{Menzel_arXiv17}.\footnote{It is also implicit in the elegant proof in \citet{Bickel_et_al_AS11}.}
It can be derived by first projecting $S_{n}\left(\theta\right)$
on to $\mathbf{A}=\left[A_{i}\right]_{1\leq i\leq N}$, $\mathbf{W}=\left[W_{i}\right]_{1\leq i\leq N}$,
$\mathbf{B}=\left[B_{j}\right]_{1\leq j\leq M}$, and $\mathbf{X}=\left[X_{i}\right]_{1\leq j\leq N}$
as follows:
\begin{align}
S_{n}\left(\theta\right) & =\mathbb{E}\left[\left.S_{n}\left(\theta\right)\right|\mathbf{W},\mathbf{X},\mathbf{A},\mathbf{B}\right]+\left\{ S_{n}\left(\theta\right)-\mathbb{E}\left[\left.S_{n}\left(\theta\right)\right|\mathbf{W},\mathbf{X},\mathbf{A},\mathbf{B}\right]\right\} \nonumber \\
 & =\frac{1}{NM}\sum_{i=1}^{N}\sum_{j=1}^{M}\bar{s}_{ij}\left(\theta\right)+\frac{1}{NM}\sum_{i=1}^{N}\sum_{j=1}^{M}\left\{ s_{ij}\left(\theta\right)-\bar{s}_{ij}\left(\theta\right)\right\} .\label{eq: first_decomposition}
\end{align}
Next observe that $\frac{1}{NM}\sum_{i=1}^{N}\sum_{j=1}^{M}\bar{s}_{ij}\left(\theta\right)$
is a two sample U-Statistic, albeit one defined partially in terms
of the latent variables $A_{i}$ and $B_{j}$. Equation \eqref{eq: U1n}
corresponds to the Hajek Projection of this U-statistic onto (separately)
$\left\{ \left(W_{i}',A_{i}\right)\right\} _{i=1}^{N}$ and $\left\{ \left(X_{j}',B_{j}\right)\right\} _{j=1}^{M}$.
Equation \eqref{eq: U2n} is the usual Hajek Projection error term. 

Define $\phi_{n}=M/n$, $\bar{s}_{1ni}^{c}\overset{def}{\equiv}\bar{s}_{1i}^{c}\left(\theta_{0,n}\right),$
$\bar{s}_{1nj}^{p}\overset{def}{\equiv}\bar{s}_{1j}^{p}\left(\theta_{0,n}\right)$
and also $\bar{s}_{nij}\overset{def}{\equiv}\bar{s}_{ij}\left(\theta_{0,n}\right)$.
Similarly let $S_{n}=S_{n}\left(\theta_{0,n}\right)$ and so on. Applying
the variance operator to $S_{n}$ yields:
\begin{align}
\mathbb{V}\left(S_{n}\right)= & \mathbb{V}\left(U_{1n}\right)+\mathbb{V}\left(U_{2n}\right)+\mathbb{V}\left(V_{n}\right)\label{eq: variance_of_score}\\
= & \frac{\Sigma_{1n}^{c}}{N}+\frac{\Sigma_{1n}^{p}}{M}+\frac{1}{NM}\left[\Sigma_{2n}-\Sigma_{1n}^{c}-\Sigma_{1n}^{p}\right]+\frac{\Sigma_{3n}}{NM}\nonumber 
\end{align}
where
\begin{align}
\Sigma_{1n}^{c} & =\mathbb{E}\left[\bar{s}_{1ni}^{c}\left(\bar{s}_{1ni}^{c}\right)'\right]\thinspace\thinspace\thinspace\Sigma_{1n}^{p}=\mathbb{E}\left[\bar{s}_{1nj}^{p}\left(\bar{s}_{1nj}^{p}\right)'\right]\label{eq: SIGMA_definitions}\\
\Sigma_{2n} & =\mathbb{E}\left[\bar{s}_{nij}\bar{s}_{nij}'\right]=\mathbb{V}\left(\mathbb{E}\left[\left.s_{nij}\right|W_{i},X_{j},A_{i},B_{j}\right]\right)\nonumber \\
\Sigma_{3n} & =\mathbb{E}\left[\left\{ s_{nij}-\bar{s}_{nij}\right\} \left\{ s_{nij}-\bar{s}_{nij}\right\} '\right]=\mathbb{E}\left[\mathbb{V}\left(\left.s_{nij}\right|W_{i},X_{j},A_{i},B_{j}\right)\right].\nonumber 
\end{align}

In the \emph{dense} case $\Sigma_{1n}^{c}$, $\Sigma_{1n}^{p},$ $\Sigma_{2n}$
and $\Sigma_{3n}$ are all constant in $n$; hence the asymptotic
properties of $S_{n}$ coincide with those of $U_{1n}$. Since $U_{1n}$
is a sum of independent random variables a standard argument gives
\begin{equation}
n^{1/2}S_{n}\overset{D}{\rightarrow}\mathcal{N}\left(0,\frac{\Sigma_{1}^{c}}{1-\phi}+\frac{\Sigma{}_{1}^{p}}{\phi}\right)\label{eq: asymptotic_distribution_dense}
\end{equation}
as long as $\Sigma_{1}^{c}$ and/or $\Sigma_{1}^{p}$ are non-zero.

Under the sparse network asymptotics considered here the order of
$\Sigma_{1n}^{c}$, $\Sigma_{1n}^{p},$ $\Sigma_{2n}$ and $\Sigma_{3n}$
varies with $n$. This affects the order of the four variance terms
in \eqref{eq: variance_of_score} and, consequently, which components
of $S_{n}$ contribute to its asymptotic properties. In Appendix \ref{app: lemmas}
I show the order of the four terms in \eqref{eq: variance_of_score}
are, respectively,
\begin{align*}
\mathbb{V}\left(S_{n}\right)= & O\left(\frac{\rho_{n}^{2}}{N}\right)+O\left(\frac{\rho_{n}^{2}}{M}\right)+O\left(\frac{\rho_{n}^{2}}{MN}\right)+O\left(\frac{\rho_{n}}{MN}\right)\\
= & O\left(\left[\frac{\lambda_{0,n}^{c}}{\phi_{n}}\right]^{2}\frac{1}{\left(1-\phi_{n}\right)}\frac{1}{n^{3}}\right)+O\left(\left[\frac{\lambda_{0,n}^{c}}{\phi_{n}}\right]^{3}\frac{1}{n^{3}}\right)\\
 & +O\left(\left[\frac{\lambda_{0,n}^{c}}{\phi_{n}}\right]^{2}\frac{1}{\phi_{n}\left(1-\phi_{n}\right)}\frac{1}{n^{4}}\right)+O\left(\frac{\lambda_{0,n}^{c}}{\phi_{n}^{2}\left(1-\phi_{n}\right)}\frac{1}{n^{3}}\right).
\end{align*}
Since $\Sigma_{1}^{c}$ and $\Sigma_{1}^{p}$ are both $O\left(\rho_{n}^{2}\right)=O\left(n^{-2}\right)$
we can multiply them by $n^{2}$ to stabilize them. Define $\tilde{\Sigma}_{1}^{c}$
to be the limit of $n^{2}\Sigma_{1n}$ and $\tilde{\Sigma}_{1}^{p}$
to be the limit of $n^{2}\Sigma_{1n}^{p}$. Similarly we can define
$\tilde{\Sigma}_{3}$ to be the limit of $n\Sigma_{3n}$, all as $n\rightarrow\infty$.
Normalizing \eqref{eq: variance_of_score} by $n^{3/2}$ therefore
gives
\begin{align}
\mathbb{V}\left(n^{3/2}S_{n}\right)= & \frac{\tilde{\Sigma}_{1}^{c}}{1-\phi}+\frac{\tilde{\Sigma}_{1}^{p}}{\phi}+\frac{\tilde{\Sigma}_{3}}{\phi\left(1-\phi\right)}+O\left(n^{-1}\right)\label{eq: asymptotic_variance_of_score}
\end{align}
where I also use the fact that $\Sigma_{2n}=O\left(n^{-2}\right)$.

Under sparse network asymptotics both $U_{1n}$ and $V_{n}$ matter.
In Appendix \ref{app: proof_of_main_result} I show that $U_{1n}+V_{n}$
is a martingale difference sequence (MDS) to which a martingale CLT
can be applied; Theorem \ref{thm: asymptotic_distribution} then follows.
\begin{thm}
\label{thm: asymptotic_distribution}Under Assumptions \ref{ass: sampling},
\ref{ass: logit_regression} and \ref{ass: identification}
\[
\sqrt{n}\left(\hat{\theta}_{n}-\theta_{n}\right)\overset{D}{\rightarrow}\mathcal{N}\left(0,\Gamma_{0}^{-1}\left[\frac{\tilde{\Sigma}_{1}^{c}}{1-\phi}+\frac{\tilde{\Sigma}_{1}^{p}}{\phi}+\frac{\tilde{\Sigma}_{3}}{\phi\left(1-\phi\right)}\right]\Gamma_{0}^{-1}\right)
\]
as $n\rightarrow\infty$.
\end{thm}
Theorem \ref{thm: asymptotic_distribution} indicates that under sparse
network asymptotics there are additional sources of sampling variation
in $\sqrt{n}\left(\hat{\theta}_{n}-\theta_{n}\right)$ relative to
those which appear in the dense case. Not incorporating these into
inference procedures will lead to tests with incorrect size and/or
confidence intervals with incorrect coverage. A further advantage
of considering sparse network asymptotics is that Theorem \ref{thm: asymptotic_distribution}
remains valid even under degeneracy of the graphon, $h_{N,M}\left(\mu,W_{i},X_{j},A_{i},B_{j},V_{ij}\right)$.
For example, if the graphon is constant in $A_{i}$ and $B_{j}$ such
that $Y_{ij}$and $Y_{ik}$ do not covary conditional on covariates
(and likewise for $Y_{ji}$ and $Y_{ki}$), then $\tilde{\Sigma}_{1}^{c}=\tilde{\Sigma}_{1}^{p}=0$,
but Theorem \ref{thm: asymptotic_distribution} nevertheless remains
valid. In contrast, under dense network asymptotics, degeneracy --
as elegantly shown by \citet{Menzel_arXiv17} -- generates additional
complications. In that case the variance of $U_{1n}$ is identically
equal to zero, while that of $U_{2n}$ and $V_{n}$ are of equal order.
In some cases, the behavior of $U_{2n}$ may even induce a non-Gaussian
limit distribution (see \citet{vanderVaart_ASBook00}). In the sparse
network cases, $U_{2n}$ is always negligible relative to $V_{n}$.
Furthermore $V_{n}$ is well approximated -- after suitable scaling
-- by a Gaussian distribution.

\section{Extensions and discussion}

In this section I connect Theorem \ref{thm: asymptotic_distribution}
to prior work on rare events logistic analysis, sketch some results
about the estimation of aggregate and average effects, and, finally,
close with a few ideas about possible areas of additional research.

\subsection*{Rare events with iid data}

\citet{King_Zeng_PA2001} discuss, with a focus on finite sample bias,
the behavior of logistic regression under ``rare events'' with iid
data. Evidently binary choice analyses where the marginal frequency
of positive events is quite small are common in empirical work.\footnote{The \citet{King_Zeng_PA2001} has close to four thousand citations
on Google Scholar.} The properties of logistic regression under sequences where the number
of ``events'' becomes small (i.e., ``rare'') relative to the sample
size as it grows were recently characterized by \citet{HaiYing_ICML2020}
(see also \citet{Owen_JMLR2007}). The main result in \citet{HaiYing_ICML2020}
coincides with a special case Theorem \ref{thm: asymptotic_distribution}
above. To see this observe that if the graphon is constant in $A_{i}$
and $B_{j}$, then $\bar{s}_{nij}$ will be identically equal to zero
for all $1\leq i\leq N$ and $1\leq j\leq M$. In this scenario there
is no ``dyadic dependence'' (after conditioning on $W_{i}$ and
$X_{j}$) and $\tilde{\Sigma}_{1}^{c}=\tilde{\Sigma}_{1}^{p}=0$.
Inspection of the calculations in Appendix \ref{app: lemmas} also
reveals that, in this case, we further have an information-matrix
type equality of $\tilde{\Sigma}_{3}=\Gamma_{0}$. Under these conditions
Theorem 1 specializes to
\[
\sqrt{n}\left(\hat{\theta}-\theta_{n}\right)\overset{D}{\rightarrow}\mathcal{N}\left(0,\Gamma_{0}^{-1}\right),
\]
as $n\rightarrow\infty$. This is precisely, up to some small differences
in notation, the result given in Theorem 1 of \citet{HaiYing_ICML2020}.\footnote{\citet{HaiYing_ICML2020} scales by the square root of the number
of events or ``ones'' in the dataset. This is, of course, of the
same order as $n$ as defined here. This difference leads to a minor
difference in our two expressions for $\Gamma_{0}$. Making these
adjustments the results coincide.}

In his analysis \citet{HaiYing_ICML2020} emphasizes that information
accumulates more slowly under ``rare event asymptotics''. In the
present setting this is reflected in the need to rescale the Hessian
matrix by $n$ to achieve convergence (see Lemma \ref{lem: Hessian_Convergence}
in Appendix \ref{app: lemmas}). In the network setting dyadic dependence
additionally slows down the rate of convergence \citep[cf.,][]{Graham_Niu_Powell_WP2019}.
If a researcher is working with a sparse network and concerned about
dyadic dependence, then she should base inference on Theorem \ref{thm: asymptotic_distribution}.
If the graphon is degenerate or, more strongly, the elements of $\left[Y_{ij}\right]_{1\leq i\leq N,1\leq j\leq M}$
are, in fact, iid, then her inferences will remain valid (since Theorem
\ref{thm: asymptotic_distribution} specializes to the ``rare events''
result of \citet{HaiYing_ICML2020} in that case).

\subsection*{Aggregate effects}

Define $e_{ni}\left(x\right)=e\left(\alpha_{0,n}+z\left(W_{i},x\right)'\beta_{0}\right)$
and $\hat{e}_{ni}\left(x\right)=e\left(\hat{\alpha}_{n}+z\left(W_{i},x\right)'\hat{\beta}\right)$
and consider an estimate of total unit sales for a product with attribute
vector $X_{j}=x$ of

\[
\hat{\gamma}_{n}\left(x\right)=\sum_{i=1}^{N}\hat{e}_{ni}\left(x\right).
\]
Under dense network asymptotics this statistic would diverge as $n\rightarrow\infty$.
Under sparse network asymptotics the sum $\sum_{i=1}^{N}e_{ni}\left(x\right)$
behaves like an average because its summands are $O\left(N^{-1}\right)$.
Consequently, $\hat{\gamma}_{n}\left(x\right)$ has a well-defined
probability limit of
\begin{equation}
\gamma_{0}\left(x\right)=\underset{n\rightarrow\infty}{\lim}\sum_{i=1}^{N}e_{ni}\left(x\right)=\left(1-\phi\right)\alpha_{0}\mathbb{E}\left[\exp\left(z\left(W_{i},x\right)'\beta_{0}\right)\right].\label{eq: conditional_average_product degree}
\end{equation}
This probability limit reflects the boundedness of average product
degree in sparse networks (i.e., in expectation, total sales of a
product are finite, even asymptotically). This result also holds within
a sub-population of products with characteristics $X_{j}=x$. Parameter
\eqref{eq: conditional_average_product degree} corresponds to a conditional
version of $\lambda_{0,n}^{p}$, the average product degree parameter
defined in Section \ref{sec: population} above.

To derive the rate of convergence and limit distribution of $\hat{\gamma}_{n}\left(x\right)$
I proceed in the usual way. A mean value expansion and Theorem \eqref{thm: asymptotic_distribution}
together yield
\begin{align}
\sqrt{n}\left(\hat{\gamma}_{n}\left(x\right)-\gamma_{0}\left(x\right)\right)\approx & \sqrt{n}\sum_{i=1}^{N}\left\{ e_{ni}\left(x\right)-\gamma_{0}\left(x\right)\right\} \nonumber \\
 & +\sum_{i=1}^{N}\bar{e}_{ni}\left(x\right)\left[1-\bar{e}_{ni}\left(x\right)\right]\left(\begin{array}{cc}
1 & z\left(W_{i},x\right)'\end{array}\right)\Gamma_{0}^{-1}\label{eq: total_sales_mean_value}\\
 & \times n^{3/2}S_{n}\left(\theta_{0,n}\right).\nonumber 
\end{align}
By the conditional mean zero property of the score function, the two
terms in \eqref{eq: total_sales_mean_value} are asymptotically uncorrelated.
The variance of the first term in \eqref{eq: total_sales_mean_value}
is
\[
\mathbb{V}\left(\sqrt{n}\sum_{i=1}^{N}\left\{ e_{ni}\left(x\right)-\gamma_{0}\left(x\right)\right\} \right)=O\left(n^{2}\left(1-\phi_{n}\right)\rho_{n}^{2}\right)=O\left(1\right),
\]
whereas the Jacobian in the second term has the approximation
\begin{multline*}
\sum_{i=1}^{N}\bar{e}_{ni}\left(x\right)\left[1-\bar{e}_{ni}\left(x\right)\right]\left(\begin{array}{cc}
1 & z\left(W_{i},x\right)'\end{array}\right)=\\
\frac{\alpha_{0}\left(1-\phi_{n}\right)}{N}\sum_{i=1}^{N}\exp\left(z\left(W_{i},x\right)'\beta_{0}\right)\left(\begin{array}{cc}
1 & z\left(W_{i},x\right)'\end{array}\right)+O_{p}\left(n^{-1}\right).
\end{multline*}
Defining $\Phi_{0}\left(x\right)\overset{def}{\equiv}$$\alpha_{0}\left(1-\phi\right)\left(\begin{array}{cc}
\mathbb{E}\left[\exp\left(z\left(W_{i},x\right)'\beta_{0}\right)\right] & \mathbb{E}\left[\exp\left(z\left(W_{i},x\right)'\beta_{0}\right)z\left(W_{i},x\right)\right]\end{array}'\right)$, $\Lambda_{0}\left(x\right)=\underset{n\rightarrow\infty}{\lim}n^{2}\left(1-\phi\right)\mathbb{V}\left(e_{ni}\left(x\right)-\gamma_{0}\left(x\right)\right)$,
and $\Omega_{0}=\Gamma_{0}^{-1}\left[\frac{\tilde{\Sigma}_{1}^{c}}{1-\phi}+\frac{\tilde{\Sigma}_{1}^{p}}{\phi}+\frac{\tilde{\Sigma}_{3}}{\phi\left(1-\phi\right)}\right]\Gamma_{0}^{-1}$
suggest that
\[
\sqrt{n}\left(\hat{\gamma}_{n}\left(x\right)-\gamma_{0}\left(x\right)\right)\rightarrow N\left(0,\Lambda_{0}\left(x\right)+\Phi_{0}\left(x\right)\Omega_{0}\Phi_{0}\left(x\right)'\right)
\]
as $n\rightarrow\infty$. Aggregate effects are estimable with the
same degree of precision as the logit coefficients themselves.

\subsection*{Average partial effects}

Next consider estimating the \emph{average} marginal effect of unit
increases in the elements of $Z_{ij}$ on making a purchase:
\begin{equation}
\hat{\gamma}_{n}=\frac{1}{NM}\sum_{i=1}^{N}\sum_{j=1}^{M}\hat{e}_{nij}\left[1-\hat{e}_{nij}\right]Z_{ij}.\label{eq: avg_marginal_effects}
\end{equation}
Recall that $e_{nij}=e\left(\alpha_{0,n}+Z_{ij}\beta_{0}\right)$
and $\hat{e}_{nij}=e\left(\hat{\alpha}_{n}+Z_{ij}'\hat{\beta}\right)$.
Interest in average partial effects of this type is widespread in
modern micro-econometric empirical research \citep[cf.,][]{Blundel_Powell_WC03,Wooldridge_IIEM05}.
Since \eqref{eq: avg_marginal_effects} \emph{is} an average of summands,
each of which is $O\left(n^{-1}\right)$, we might expect some variance
reduction relative to the aggregate case just discussed. In a certain
sense, this conjecture appears to be correct. 

Define $\gamma_{0,n}=\mathbb{E}_{N,M}\left[e_{nij}\left(1-e_{nij}\right)Z_{ij}\right]=O\left(\rho_{n}\right)$;
a mean-value expansion and some re-scaling yields
\begin{align}
\rho_{n}n^{3/2}\left(\frac{\hat{\gamma}_{n}-\gamma_{0,n}}{\rho_{n}}\right) & =n^{3/2}\left\{ \frac{1}{NM}\sum_{i=1}^{N}\sum_{j=1}^{M}\left[e_{nij}\left(1-e_{nij}\right)Z_{ij}-\gamma_{0,n}\right]\right\} \nonumber \\
 & +n\left\{ \frac{1}{NM}\sum_{i=1}^{N}\sum_{j=1}^{M}\bar{e}_{nij}\left[1-\bar{e}_{nij}\right]\left[1-2\bar{e}_{nij}\right]\left[\begin{array}{cc}
Z_{ij} & Z_{ij}Z_{ij}'\end{array}\right]\right\} \label{eq: APE}\\
 & \times n^{1/2}\left(\hat{\theta}_{n}-\theta_{0,n}\right).\nonumber 
\end{align}
As above, the conditional mean zero property of the score function
ensures that the two terms in \eqref{eq: APE} are asymptotically
uncorrelated. We rescale the estimate and parameter using $\rho_{n}$
since $\gamma_{0,n}\rightarrow0$ as $n\rightarrow\infty$ \citep[cf.,][]{Bickel_et_al_AS11}.
The need to rescale the Jacobian in \eqref{eq: APE} stems from the
observation that
\begin{multline*}
n\left\{ \frac{1}{NM}\sum_{i=1}^{N}\sum_{j=1}^{M}\bar{e}_{nij}\left[1-\bar{e}_{nij}\right]\left[1-2\bar{e}_{nij}\right]\left[\begin{array}{cc}
Z_{ij} & Z_{ij}Z_{ij}'\end{array}\right]\right\} =\\
\frac{1}{NM}\sum_{i=1}^{N}\sum_{j=1}^{M}\alpha_{0}\exp\left(Z_{ij}'\beta_{0}\right)\left[\begin{array}{cc}
Z_{ij} & Z_{ij}Z_{ij}'\end{array}\right]+O_{p}\left(n\rho_{n}^{2}\right).
\end{multline*}

Next observe that first term in \eqref{eq: APE} is a two-sample U-Statistics.
A Hoeffding variance decomposition gives
\[
\mathbb{V}\left(\frac{1}{NM}\sum_{i=1}^{N}\sum_{j=1}^{M}e_{nij}\left(1-e_{nij}\right)Z_{ij}\right)=\frac{\Lambda_{1n}^{c}}{N}+\frac{\Lambda_{1n}^{p}}{M}+\frac{1}{NM}\left[\Lambda_{2n}-\Lambda_{1n}^{c}-\Lambda_{1n}^{p}\right]
\]
with
\begin{align*}
\Lambda_{1n}^{c} & =\mathbb{C}\left(e_{nij}\left(1-e_{nij}\right)Z_{ij},e_{nik}\left(1-e_{nik}\right)Z_{ik}'\right)\\
\Lambda_{1n}^{p} & =\mathbb{C}\left(e_{nji}\left(1-e_{nji}\right)Z_{ji},e_{nki}\left(1-e_{nki}\right)Z_{ki}'\right)\\
\Lambda_{2n} & =\mathbb{V}\left(e_{nij}\left(1-e_{nij}\right)Z_{ij}\right).
\end{align*}
Inspection indicates that $\Lambda_{1n}^{c}=O\left(\rho_{n}^{2}\right)$,
$\Lambda_{1n}^{p}=O\left(\rho_{n}^{2}\right)$ and $\Lambda_{2n}^{c}=O\left(\rho_{n}^{2}\right)$
and hence that
\[
n^{3}\mathbb{V}\left(\frac{1}{NM}\sum_{i=1}^{N}\sum_{j=1}^{M}e_{nij}\left(1-e_{nij}\right)Z_{ij}\right)=\frac{\tilde{\Lambda}_{1}^{c}}{1-\phi}+\frac{\tilde{\Lambda}_{1}^{p}}{\phi}+O\left(n^{-1}\right).
\]

Putting these calculations together suggests that

\[
\rho_{n}n^{3/2}\left(\frac{\hat{\gamma}_{n}-\gamma_{0,n}}{\rho_{n}}\right)\overset{D}{\rightarrow}\mathcal{N}\left(\frac{\tilde{\Lambda}_{1}^{c}}{1-\phi}+\frac{\tilde{\Lambda}_{1}^{p}}{\phi}+\Phi_{0}\Omega_{0}\Phi_{0}'\right)
\]
with $\Phi_{0}\overset{def}{\equiv}\alpha_{0}\mathbb{E}\left[\exp\left(Z_{ij}'\beta_{0}\right)\left[\begin{array}{cc}
Z_{ij} & Z_{ij}Z_{ij}'\end{array}\right]\right]$.

If we set $T=NM=O\left(n^{2}\right)$, then we have that $T^{1/4}\left(\hat{\theta}_{n}-\theta_{0,n}\right)$
has a Gaussian limit distribution. The rate of convergence of $\hat{\theta}_{n}$
toward $\theta_{0,n}$ is slow. For average partial effects we need
to rescale in order ensure a meaningful probability limit. Let $\gamma_{0,n}^{*}=\gamma_{0,n}/\rho_{n}$
and similarly for $\hat{\gamma}_{n}$; the result above implies that
$T^{1/4}\left(\hat{\gamma}_{n}^{*}-\gamma_{0,n}^{*}\right)$ is also
Gaussian. In this sense the rates-of-convergence for the logit coefficients
and their APEs coincide. However, if we think in terms of the resulting
implied approximation to the finite sample distribution of the two
parameter estimates, we have $\mathbb{V}\left(\hat{\theta}_{n}\right)=O\left(T^{-1/2}\right)=O\left(n^{-1}\right)$,
but $\mathbb{V}\left(\hat{\gamma}_{n}\right)=O\left(T^{-3/2}\right)=O\left(n^{-3}\right)$.
In this sense inference on APEs appears to be more precise.

\subsection*{Areas for additional research}

For empirical researchers the main implication of this paper is to
use an estimate for the variance of $S_{N}$ that includes all components
-- even ones that are negligible under certain asymptotic sequences
-- when constructing standard errors. This is not a new idea. In
the context of U-statistics it goes back to \citet{Hoeffding_AMS48}.
It is implicit in \citet{Holland_Leinhardt_SM76} in their work on
subgraph counts; see also the recent work on dyadic regression by
\citet{Fafchamp_Gubert_JDE07}, \citet{Cameron_Miller_WP14} and \citet{Aronow_et_al_PA17},
as well as that on density weighted average derivatives by \citet{Cattaneo_et_al_ET14}.
However, the small amount of extant formal limit theory for dyadic
regression (cited earlier) suggests different approaches to variance
estimation. This paper has outlined an asymptotic framework that provides
formal justification for one of the leading ``practical'' approaches
to inference in the presence of dyadic dependence. \citet{Graham_Book_DR_Chap2020}
discusses variance estimation for dyadic regression in detail, advocating
a variant of the estimate proposed by \citet{Fafchamp_Gubert_JDE07},
\citet{Cameron_Miller_WP14} and \citet{Aronow_et_al_PA17}. Theorem
\ref{thm: asymptotic_distribution} provides a formal justification
for this recommendation.

Many outstanding questions remain. Can the above framework be generalized
to a generic dyadic M-estimation setting? What is the ``general''
notion of ``sparseness'' needed for this? Extensions to semiparametric
regression models are also of interest. In recent work, \citet{Menzel_arXiv17}
and \citet{Davezies_et_al_AS20} propose bootstrap procedures for
dyadic regression. Are these procedures also valid under sparse network
asymptotics and, if not, how might they be adapted to be so? The aggregate
and average effect examples sketched above suggest that the systematic
exploration of policy analysis questions -- considered under dense
network asymptotics by \citet{Graham_HBE2020} -- would be interesting.
Finally, although it seems likely that -- in the absence of imposing
more structure -- that the composite maximum likelihood estimator
is efficient, this is currently only a conjecture. 

\appendix
\begin{center}
\textbf{\Large{}Appendix}{\Large\par}
\par\end{center}

The appendix includes proofs of the theorems stated in the main text
as well as statements and proofs of supplemental lemmata -- called
limonata here. All notation is as established in the main text unless
stated otherwise. Equation number continues in sequence with that
established in the main text.

\section{\label{app: lemmas}Preliminary lemmata and proofs}

Let $R_{ij}=\left(1,Z_{ij}'\right)'$ and note that, for $e\left(v\right)=\exp\left(v\right)/\left[1+\exp\left(v\right)\right],$
we have that $e'\left(v\right)=e\left(v\right)\left[1-e\left(v\right)\right]$
and $e''\left(v\right)=e\left(v\right)\left[1-e\left(v\right)\right]\left[1-2e\left(v\right)\right]$.
With this notation we can write the first three derivatives of the
kernel function of the composite log-likelihood with respect $\theta_{n}$
as
\begin{align}
s_{ij}\left(\theta_{n}\right) & =\left(Y_{ij}-e_{ij}\left(\theta_{n}\right)\right)R_{ij}\label{eq: 1st_der_logl}\\
\frac{\partial s_{ij}\left(\theta_{n}\right)}{\partial\theta'} & =-e_{ij}\left(\theta_{n}\right)\left[1-e_{ij}\left(\theta_{n}\right)\right]R_{ij}R_{ij}'\label{eq: 2nd_der_logl}\\
\frac{\partial}{\partial\theta'}\left\{ \frac{\partial s_{ij}\left(\theta_{n}\right)}{\partial\theta_{p}}\right\}  & =-e_{ij}\left(\theta_{n}\right)\left[1-e_{ij}\left(\theta_{n}\right)\right]\left[1-2e_{ij}\left(\theta_{n}\right)\right]R_{ij}R_{ij}'R_{p,ij}\label{eq: 3rd_der_logl}
\end{align}
with \eqref{eq: 3rd_der_logl} holding for $\text{for \ensuremath{p=1,\ldots,\dim\left(\theta_{n}\right)}.}$

Let $\mathbf{t}=\left(\theta_{n}-\theta_{0,n}\right)$ and recall
that $\alpha_{n}=\ln\left(\alpha/n\right)$ and $\alpha_{0,n}=\ln\left(\alpha_{0}/n\right)$.
This implies that $\mathbf{t}=\left(\ln\left(\alpha/\alpha_{0}\right),\left(\beta-\beta_{0}\right)'\right)'$
does not vary with $n$ and hence that $\mathbf{t}\in\mathbb{T}$
with $\mathbb{T}$ compact by Assumption \ref{ass: logit_regression}.
Associated with any $\mathbf{t}\in\mathbb{T}$ is a $\theta_{n}\in\Theta_{n}$;
furthermore associated with this $\theta_{n}$ is a $\theta\in\Theta$.
With these preliminaries we can show that $nH_{n}\left(\theta_{n}\right)$
converges uniformly to $\Gamma\left(\theta\right)$, as defined in
equation \eqref{eq: GAMMA} of the main text.
\begin{lem}
\label{lem: Hessian_Convergence}\textsc{(Uniform Hessian Convergence)
}Under Assumptions \ref{ass: sampling}, \ref{ass: logit_regression}
and \ref{ass: identification} 
\[
\underset{\theta\in\Theta}{\sup}\left\Vert nH_{n}\left(\theta_{n}\right)-\Gamma\left(\theta\right)\right\Vert \overset{p}{\rightarrow}0.
\]
\end{lem}
\begin{proof}
Recall that $\theta_{n}=\theta_{0,n}+\mathbf{t}$ and hence that $H_{n}\left(\theta_{0,n}+\mathbf{t}\right)=H_{n}\left(\theta_{n}\right).$
The mean value theorem, as well as compatibility of the Frobenius
matrix norm with the Euclidean vector norm, gives for any $\mathbf{t}$
and $\bar{\mathbf{t}}$ both in $\mathbb{T}$, 
\begin{align*}
\left\Vert H_{n}\left(\theta_{0,n}+\mathbf{t}\right)-H_{n}\left(\theta_{0,n}+\bar{\mathbf{t}}\right)\right\Vert _{2,1} & \leq\sum_{p=1}^{\dim\left(\theta_{n}\right)}\left\Vert \frac{1}{NM}\sum_{i=1}^{N}\sum_{j=1}^{M}\frac{\partial}{\partial\theta'}\left\{ \frac{\partial s_{ij}\left(\theta_{0,n}+\mathbf{t}\right)}{\partial\theta_{p}}\right\} \right\Vert _{F}\left\Vert \mathbf{t}-\bar{\mathbf{t}}\right\Vert _{2}.
\end{align*}
Since $\mathbb{E}\left[e_{ij}\left(\theta_{n}\right)\left[1-e_{ij}\left(\theta_{n}\right)\right]\left[1-2e_{ij}\left(\theta_{n}\right)\right]\right]=O\left(\rho_{n}\right)=O\left(n^{-1}\right)$
we have that, inspecting \eqref{eq: 2nd_der_logl} above, for any
$\mathbf{t}\in\mathbb{T}$,
\begin{align*}
\left\Vert \frac{1}{NM}\sum_{i=1}^{N}\sum_{j=1}^{M}\frac{\partial}{\partial\theta'}\left\{ \frac{\partial s_{ij}\left(\theta_{0,n}+\mathbf{t}\right)}{\partial\theta_{p}}\right\} \right\Vert _{F} & =O_{p}\left(n^{-1}\right).
\end{align*}
This gives $\left\Vert nH_{N}\left(\theta_{0,n}+\mathbf{t}\right)-nH_{n}\left(\theta_{0,n}+\bar{\mathbf{t}}\right)\right\Vert _{2,1}\leq O_{p}\left(1\right)\cdot\left\Vert \mathbf{t}-\bar{\mathbf{t}}\right\Vert _{2}$.
Next, again recalling that $\theta_{0,n}+\mathbf{t}=\theta_{n}$,
we have that
\begin{align*}
H_{n}\left(\theta_{0,n}+\mathbf{t}\right) & =-\frac{1}{NM}\sum_{i=1}^{N}\sum_{j=1}^{M}e_{ij}\left(\theta_{n}\right)\left[1-e_{ij}\left(\theta_{n}\right)\right]R_{ij}R_{ij}'\\
 & =-\frac{1}{NM}\sum_{i=1}^{N}\sum_{j=1}^{M}\frac{\alpha}{n}\exp\left(R_{ij}'\beta\right)R_{ij}R_{ij}'+O_{p}\left(\frac{1}{n^{2}}\right),
\end{align*}
which gives, using a law of large numbers for U-Statistics, $nH_{n}\left(\theta_{n}\right)\overset{p}{\rightarrow}\Gamma\left(\theta\right)$
for all $\mathbf{t}\in\mathbb{T}$. The claim then follows from an
application of Lemma 2.9 of \citet[p. 2138]{Newey_McFadden_HBE94}.
\end{proof}

\section{\label{app: proof_of_main_result}Proof of Theorem \ref{thm: asymptotic_distribution}}

\subsection*{Asymptotic variance of the score}

To prove \eqref{eq: variance_of_score}, the decomposition of the
variance of the score given in the main text, and hence that
\[
\mathbb{V}\left(n^{3/2}S_{n}\right)=\frac{\tilde{\Sigma}_{1}^{c}}{1-\phi}+\frac{\tilde{\Sigma}_{1}^{p}}{\phi}+\frac{\tilde{\Sigma}_{3}}{\phi\left(1-\phi\right)}+O\left(n^{-1}\right)
\]
Let $e_{ij}=e\left(\alpha_{0,n}+Z_{ij}'\beta_{0}\right)$; using the
definitions given in \eqref{eq: SIGMA_definitions} of the main text
we have that
\begin{align}
\Sigma_{1n}^{c}= & \mathbb{E}\left[\left(Y_{12}-e_{12}\right)\left(Y_{13}-e_{13}\right)R_{12}R_{13}'\right]\nonumber \\
= & O\left(\rho_{n}^{2}\right)\label{eq: order_SIGMA1_c}
\end{align}
and also that
\begin{align}
\Sigma_{1n}^{p}= & \mathbb{E}\left[\left(Y_{21}-e_{21}\right)\left(Y_{31}-e_{31}\right)R_{21}R_{31}'\right].\nonumber \\
= & O\left(\rho_{n}^{2}\right).\label{eq: order_SIGMA1_p}
\end{align}

Turning to $\Sigma_{2n}$ and $\Sigma_{3n}$ we get that
\begin{align}
\Sigma_{2n}= & \mathbb{E}\left[\mathbb{E}\left[\left.\left(Y_{12}-e_{12}\right)R_{21}\right|W_{1},X_{2},A_{1},B_{2}\right]\right.\nonumber \\
 & \left.\times\mathbb{E}\left[\left.\left(Y_{12}-e_{12}\right)R_{21}\right|W_{1},X_{2},A_{1},B_{2}\right]'\right]\nonumber \\
= & O\left(\rho_{n}^{2}\right)\label{eq: order_SIGMA2}
\end{align}
and that
\begin{align}
\Sigma_{3n}= & \mathbb{E}\left[\left\{ s_{nij}-\bar{s}_{nij}\right\} \left\{ s_{nij}-\bar{s}_{nij}\right\} '\right]\nonumber \\
= & O\left(\rho_{n}\right)\label{eq: order_SIGMA3}
\end{align}
by virtue of the equality $Y_{ij}^{2}=Y_{ij}$(which holds because
$Y_{ij}$ is binary-valued). 

From \eqref{eq: sparse_prob} we have that $\rho_{n}=O\left(n^{-1}\right)$,
hence \eqref{eq: order_SIGMA1_c} implies that $n^{2}\Sigma_{1n}^{c}=O\left(1\right)$,
\eqref{eq: order_SIGMA1_p} that $n^{2}\Sigma_{1n}^{p}=O\left(1\right)$,
and \eqref{eq: order_SIGMA3} that $n^{2}\Sigma_{1n}^{p}=O\left(1\right)$.
This gives

\begin{align*}
\mathbb{V}\left(S_{n}\right)= & O\left(\frac{\rho_{n}^{2}}{N}\right)+O\left(\frac{\rho_{n}^{2}}{M}\right)+O\left(\frac{\rho_{n}^{2}}{MN}\right)+O\left(\frac{\rho_{n}}{MN}\right)\\
= & O\left(\left[\frac{\lambda_{0,n}^{c}}{M}\right]^{2}\frac{1}{N}\right)+O\left(\left[\frac{\lambda_{0,n}^{c}}{M}\right]^{3}\right)+O\left(\left[\frac{\lambda_{0,n}^{c}}{M}\right]^{2}\frac{1}{MN}\right)+O\left(\frac{\lambda_{0,n}^{c}}{M}\frac{1}{MN}\right)\\
= & O\left(\left[\frac{\lambda_{0,n}^{c}}{\phi_{n}}\right]^{2}\frac{1}{\left(1-\phi_{n}\right)}\frac{1}{n^{3}}\right)+O\left(\left[\frac{\lambda_{0,n}^{c}}{\phi_{n}}\right]^{3}\frac{1}{n^{3}}\right)\\
 & +O\left(\left[\frac{\lambda_{0,n}^{c}}{\phi_{n}}\right]^{2}\frac{1}{\phi_{n}\left(1-\phi_{n}\right)}\frac{1}{n^{4}}\right)+O\left(\frac{\lambda_{0,n}^{c}}{\phi_{n}^{2}\left(1-\phi_{n}\right)}\frac{1}{n^{3}}\right)\\
= & O\left(n^{3}\right)+O\left(n^{3}\right)+O\left(n^{4}\right)+O\left(n^{3}\right),
\end{align*}
as needed.

\subsection*{Triangular array setup}

Consider the following triangular array $\left\{ Z_{nt}\right\} $:
\begin{align*}
Z_{n1} & =\frac{1}{N}\bar{s}_{1n1}^{c}\\
 & \vdots\\
Z_{nN} & =\frac{1}{N}\bar{s}_{1nN}^{c}\\
Z_{nN+1} & =\frac{1}{M}\bar{s}_{1n1}^{p}\\
 & \vdots\\
Z_{nN+M} & =\frac{1}{M}\bar{s}_{1nM}^{p}\\
Z_{nN+M+1} & =\frac{1}{NM}\left(s_{n11}-\bar{s}_{n11}\right)\\
 & \vdots\\
Z_{nN+M+NM} & =\frac{1}{NM}\left(s_{nNM}-\bar{s}_{nNM}\right),
\end{align*}
with $T=T\left(n\right)=N+M+NM$. For any vector $X_{i},$ let $X_{1}^{t}=\left(X_{1},\ldots,X_{t}\right)'$.
Iterated expectations, as well as the conditional independence relationships
implied by dyadic dependence (Assumptions \ref{ass: sampling} and
\ref{ass: logit_regression}), yield
\begin{align*}
\mathbb{E}\left[\left.Z_{ni}\right|Z_{n1}^{i-1}\right] & =0,
\end{align*}
establishing that $\left\{ Z_{Ni}\right\} $ is a martingale difference
sequence (MDS). The variance of this MDS is
\begin{align*}
\bar{\Delta}_{n} & \overset{def}{\equiv}\mathbb{V}\left(\sum_{t=1}^{T}Z_{ni}\right)\\
 & =\frac{\Sigma_{1n}^{c}}{N}+\frac{\Sigma_{1n}^{p}}{M}+\frac{\Sigma_{3n}}{NM}.
\end{align*}

To show asymptotic normality of $N^{3/2}S_{n}\left(\theta_{0,n}\right)$
I first show, recalling decomposition \eqref{eq: score_decomposition}
in the main test, that, for a vector of constants $c,$
\begin{equation}
\left(c'\bar{\Delta}_{n}c\right)^{-1/2}c'S_{n}=\left(c'\bar{\Delta}_{n}c\right)^{-1/2}c'\left[U_{1n}+V_{n}\right]+o_{p}\left(1\right)\label{eq: U2N_is_op1}
\end{equation}
and subsequently that
\begin{equation}
\left(c'\bar{\Delta}_{n}c\right)^{-1/2}c'\left[U_{1n}+V_{n}\right]\overset{p}{\rightarrow}\mathcal{N}\left(0,1\right).\label{eq: asym_norm}
\end{equation}

To show \eqref{eq: U2N_is_op1} observe that
\begin{align*}
c'\bar{\Delta}_{n}c & =O\left(\frac{\rho_{n}^{2}}{N}+\frac{\rho_{n}^{2}}{M}+\frac{\rho_{n}}{NM}\right)\\
 & =O\left(\frac{\rho_{n}^{2}}{n}\left(\frac{1}{1-\phi_{n}}+\frac{1}{\phi_{n}}+\frac{1}{\left(1-\phi_{n}\right)\lambda_{n}^{c}}\right)\right)\\
 & =O\left(\frac{\rho_{n}^{2}}{n}\right)
\end{align*}
and hence that $\left(c'\bar{\Delta}_{N}c\right)^{-1}=O\left(n\rho_{n}^{-2}\right)$
as long as $\lambda_{n}^{c}\geq C>0$ and $\phi\in\left(0,1\right)$
(see Assumptions \ref{ass: sampling} and \ref{ass: logit_regression}).
Additionally using \eqref{eq: order_SIGMA2} yields
\begin{align*}
\left(c'\bar{\Delta}_{n}c\right)^{-1/2}c'U_{2n} & =O_{p}\left(n^{1/2}\rho_{n}^{-1}\right)O_{p}\left(\rho_{n}^{2}\right)\\
 & =O_{p}\left(n^{1/2}\rho_{n}\right)\\
 & =o_{p}\left(1\right),
\end{align*}
as long as $\rho_{n}=O\left(n^{-\alpha}\right)$ for $\alpha>\frac{1}{2}$,
as is maintained here. This proves assertion \eqref{eq: U2N_is_op1}.

\subsection*{Central limit theorem}

To show \eqref{eq: asym_norm} I verify the conditions of Corollary
5.26 of Theorem 5.24 in \citet{White_Bk01}; specifically the Lyapunov
condition, for $r>2$

\begin{equation}
\sum_{t=1}^{T\left(n\right)}\mathbb{E}\left[\left(\frac{c'Z_{nt}}{\left(c'\bar{\Delta}_{nt}c\right)}\right)^{r}\right]=o\left(1\right)\label{eq: Lypaunov}
\end{equation}
and the stability condition
\begin{equation}
\sum_{t=1}^{T\left(n\right)}\frac{\left(c'Z_{Nt}\right)^{2}}{c'\bar{\Delta}_{N}c}\overset{p}{\rightarrow}1.\label{eq: stability}
\end{equation}
I will show \eqref{eq: Lypaunov} for $r=3$. Observe that
\begin{align*}
\mathbb{E}\left[\left(\frac{1}{N}c'\bar{s}_{1ni}^{c}\right)^{3}\right]= & O\left(\frac{\rho_{n}^{3}}{N^{3}}\right)\\
\mathbb{E}\left[\left(\frac{1}{M}c'\bar{s}_{1ni}^{p}\right)^{3}\right]= & O\left(\frac{\rho_{n}^{3}}{M^{3}}\right)\\
\mathbb{E}\left[\left(\frac{1}{NM}c'\left(s_{n11}-\bar{s}_{n11}\right)\right)^{3}\right]= & O\left(\frac{\rho_{n}}{N^{3}M^{3}}\right)
\end{align*}
These calculations, as well as independence of summands $1$ to $N$,
$N+1$ to $N+M$ and $N+M+1$ to $N+M+NM$, imply that

\begin{align*}
\sum_{t=1}^{T\left(n\right)}\mathbb{E}\left[\left(\frac{c'Z_{Nt}}{\left(c'\bar{\Delta}_{N}c\right)}\right)^{3}\right]= & O_{p}\left(n^{3/2}\rho_{N}^{-3}\right)\left\{ O\left(\frac{\rho_{n}^{3}}{N^{2}}\right)+O\left(\frac{\rho_{n}^{3}}{M^{2}}\right)+O\left(\frac{\rho_{n}}{N^{2}M^{2}}\right)\right\} \\
= & O_{p}\left(\frac{1}{\left(1-\phi_{n}\right)^{2}n^{1/2}}\right)+O_{p}\left(\frac{1}{\phi_{n}^{2}n^{1/2}}\right)+O_{p}\left(\frac{1}{\left(1-\phi_{n}\right)^{2}\lambda_{n}^{c}n^{1/2}}\right)\\
 & O_{p}\left(n^{-1/2}\right)\\
 & o_{p}\left(1\right)
\end{align*}
as required.

To verify the stability condition \eqref{eq: stability} I re-write
it as
\begin{equation}
\sum_{t=1}^{T\left(n\right)}\frac{1}{n\left(c'\bar{\Delta}_{n}c\right)}n\left\{ \left(c'Z_{nt}\right)^{2}-\mathbb{E}\left[\left(c'Z_{nt}\right)^{2}\right]\right\} \overset{p}{\rightarrow}0\label{eq: new_stability}
\end{equation}
Since $n\left(c'\bar{\Delta}_{N}c\right)^{-1}=O\left(n\cdot n\rho_{N}^{-2}\right)=O\left(1\right)$
the stability condition \eqref{eq: stability} will hold if the numerator
in \eqref{eq: new_stability} -- $\sum_{t=1}^{T\left(n\right)}n\left\{ \left(c'Z_{nt}\right)^{2}-\mathbb{E}\left[\left(c'Z_{nt}\right)^{2}\right]\right\} $
-- converges in probability to zero. Expanding the square we get
that
\[
\mathbb{E}\left[\left(n\left\{ \left(c'Z_{nt}\right)^{2}-\mathbb{E}\left[\left(c'Z_{nt}\right)^{2}\right]\right\} \right)^{2}\right]=n^{2}\left\{ \mathbb{E}\left[\left(c'Z_{nt}\right)^{4}\right]-\left(\mathbb{E}\left[\left(c'Z_{nt}\right)^{2}\right]\right)^{2}\right\} .
\]
We then have
\[
\mathbb{E}\left[\left(c'Z_{nt}\right)^{2}\right]=\left\{ \begin{array}{ll}
\frac{1}{N^{2}}c'\Sigma_{1n}^{c}c=O\left(\left[\frac{\lambda_{n}^{c}}{\phi_{n}\left(1-\phi_{n}\right)}\right]^{2}\frac{1}{n^{4}}\right), & t=1,\ldots,N\\
\frac{1}{M^{2}}c'\Sigma_{1n}^{p}c=O\left(\left[\frac{\lambda_{n}^{c}}{\phi_{n}^{2}}\right]^{2}\frac{1}{n^{4}}\right), & t=N+1,\ldots,N+M\\
\frac{1}{N^{2}M^{2}}c'\Sigma_{3N}c=O\left(\frac{\lambda_{n}^{c}}{\phi_{n}^{3}\left(1-\phi_{n}\right)^{2}}\frac{1}{n^{5}}\right), & t=N+M+1,\ldots,N+M+NM
\end{array}\right.
\]
and
\[
\mathbb{E}\left[\left(c'Z_{nt}\right)^{4}\right]=\left\{ \begin{array}{ll}
\frac{\mathbb{E}\left[\left(c'\bar{s}_{1n1}^{c}\right)^{4}\right]}{N^{4}}=O\left(\frac{1}{\left(1-\phi_{n}\right)^{4}}\frac{\rho_{n}^{4}}{n^{4}}\right), & t=1,\ldots,N\\
\frac{\mathbb{E}\left[\left(c'\bar{s}_{1n1}^{p}\right)^{4}\right]}{M^{4}}=O\left(\frac{1}{\phi_{n}^{4}}\frac{\rho_{n}^{4}}{n^{4}}\right), & t=N+1,\ldots,N+M\\
\frac{\mathbb{E}\left[\left(c'\left(s_{n11}-\bar{s}_{n11}\right)\right)^{4}\right]}{N^{4}M^{4}}=O\left(\frac{1}{\phi_{n}^{4}\left(1-\phi_{n}\right)^{4}}\frac{\rho_{n}}{n^{8}}\right), & t=N+M+1,\ldots,N+M+NM
\end{array}\right..
\]
Since $T\left(n\right)=N+M+NM=O\left(n^{2}\right)$, the summands
of $\frac{1}{T\left(n\right)}\sum_{t=1}^{T\left(n\right)}T\left(n\right)n\left\{ \left(c'Z_{nt}\right)^{2}-\mathbb{E}\left[\left(c'Z_{nt}\right)^{2}\right]\right\} $
all have variances which are $O\left(n^{-2}\right)$ or smaller:
\begin{multline*}
T\left(n\right)^{2}n^{2}\left\{ \mathbb{E}\left[\left(c'Z_{nt}\right)^{4}\right]-\left(\mathbb{E}\left[\left(c'Z_{nt}\right)^{2}\right]\right)^{2}\right\} =\\
\left\{ \begin{array}{ll}
T\left(n\right)^{2}n^{2}\left[O\left(n^{-8}\right)+O\left(n^{-8}\right)\right]=O\left(n^{-2}\right), & t=1,\ldots,N\\
T\left(n\right)^{2}n^{2}\left[O\left(n^{-8}\right)+O\left(n^{-8}\right)\right]=O\left(n^{-2}\right), & t=N+1,\ldots,N+M\\
T\left(n\right)^{2}n^{2}\left[O\left(n^{-9}\right)+O\left(n^{-10}\right)\right]=O\left(n^{-3}\right), & t=N+M+1,\ldots,N+M+NM
\end{array}\right.
\end{multline*}
Since the summands of the numerator in \eqref{eq: new_stability}
are all mean zero with variances shrinking to zero as $n\rightarrow\infty$
condition \eqref{eq: new_stability} holds as required.

\bibliographystyle{apalike}
\bibliography{../Networks_Book/Finished/Reference_BibTex/Networks_References}

\end{document}